\documentclass[11pt]{article}
\usepackage{amsmath, amsthm}
\usepackage{latexsym}
\usepackage{graphicx}
\makeatletter
\newfont{\footsc}{cmcsc10 at 8truept}
\newfont{\footbf}{cmbx10 at 8truept}
\newfont{\footrm}{cmr10 at 10truept}
\makeatother
\newtheorem{theorem}{\bf Theorem}
\newtheorem{proposition}{\bf Proposition}
\newtheorem{lemma}{\bf Lemma}
\newtheorem{corollary}{\bf Corollary}

\begin{document}
\title{Strict Monotonicity and Convergence Rate of Titterington's Algorithm for Computing D-optimal Designs}

\author{Yaming Yu\\
\small Department of Statistics\\[-0.8ex]
\small University of California\\[-0.8ex] 
\small Irvine, CA 92697, USA\\[-0.8ex]
\small \texttt{yamingy@uci.edu}}

\date{}
\maketitle

\begin{abstract}
We study a class of multiplicative algorithms introduced by Silvey et al.\ (1978) for computing D-optimal designs.  Strict monotonicity is established for a variant considered by Titterington (1978).  A formula for the rate of convergence is also derived.  This is used to explain why modifications considered by Titterington (1978) and Dette et al. (2008) usually converge faster. 

\end{abstract}

\section{Introduction}
Let $\mathcal{X}=\{x_1, \ldots, x_n\}\subset \mathbf{R}^m$ be a design space of $n$ points ($n\geq m$).  We consider computational aspects of D-optimal design (approximate theory) for linear models (Kiefer 1974; Silvey, 1980; P\'{a}zman, 1986; Pukelsheim, 1993).  The D-criterion seeks to maximize the determinant of the $m\times m$ matrix 
$$M(w)=\sum_{i=1}^n w_ix_ix_i^\top$$
with respect to $w=(w_1,\ldots, w_n)^\top\in \bar{\Omega},$ where $\bar{\Omega}$ denotes the closure of $\Omega=\{w:\ \sum_{i=1}^n w_i=1,\ w_i> 0\}$.  As usual, $M(w)$ represents the Fisher information for the $m\times 1$ parameter $\theta$ in the linear model 
$$y|(x,\theta)\sim {\rm N}(x^\top \theta, \sigma^2)$$
when the number of units assigned to $x_i$ is proportional to $w_i$.  An iterative procedure to solve this problem (Silvey et al.\ 1978) is as follows. 

{\bf Algorithm~I} 
\begin{enumerate}
\item
Set $w^{(0)}=(w_1^{(0)}, \ldots, w_n^{(0)})^\top\in \Omega$. 
\item
For $t=1,2,\ldots$, compute 
\begin{equation}
\label{alg1}
w_i^{(t)}=w_i^{(t-1)} \frac{x_i^\top M^{-1}(w^{(t-1)}) x_i}{m},\quad i=1, \ldots, n.
\end{equation}
Iterate until convergence.
\end{enumerate}

Algorithm~II (Titterington 1978), a variant of Algorithm~I, can be applied when the design points include an intercept, i.e., $x_i=(1, z_i^\top)^\top$, where $z_i\in \mathbf{R}^{m-1}$.

{\bf Algorithm II} 
\begin{enumerate}
\item
The same as Step 1 of Algorithm~I.
\item
For $t=1,2\ldots$, compute 
$$\bar{z}=\sum_{i=1}^n w_i^{(t-1)} z_i;\quad M_c(w^{(t-1)})=\sum_{i=1}^n w_i^{(t-1)} (z_i-\bar{z}) (z_i-\bar{z})^\top;$$
\begin{equation}
\label{alg2}
w_i^{(t)}=w_i^{(t-1)} \frac{(z_i-\bar{z})^{\top} M_c^{-1}(w^{(t-1)}) (z_i-\bar{z})}{m-1},\quad i=1, \ldots, n.
\end{equation}
Iterate until convergence.
\end{enumerate} 

In a form that resembles (\ref{alg1}) even more closely, (\ref{alg2}) reads  
\begin{equation}
\label{alg2+}
w_i^{(t)}=w_i^{(t-1)} \frac{x_i^\top M^{-1}(w^{(t-1)}) x_i-\alpha}{m-\alpha},\quad i=1, \ldots, n,
\end{equation}
with $\alpha=1$.  Note that (\ref{alg2+}) does not require that the design points include an intercept, and therefore can be more broadly applicable than (\ref{alg2}).  In what follows $x_i$ need not include an intercept when we refer to (\ref{alg2+}). 

Algorithms I and II have generated considerable interest; see, for example, Titterington (1976, 1978), Silvey et al. (1978), P\'{a}zman (1986), Torsney and Mandal (2006), Harman and Pronzato (2007), Dette et al.\ (2008), and Yu (2010).  Algorithm~I is known to be monotonic (Titterington 1976), i.e., $\det M(w^{(t)})$ never decreases in $t$.  Monotonicity of Algorithm~II has been resolved recently (Titterington 1978; Yu, 2010).  Part of this work aims to extend this to strict monotonicity, thereby showing that Algorithm~II converges monotonically for $m\geq 3$, and fully resolving Titterington's (1978) conjecture. 

It has been observed that Algorithm~II usually converges faster than Algorithm~I; see, e.g., Dette et al. (2008).  Another goal of this work is to give an explanation of this by a theoretical analysis of the convergence rates.  Our investigation is partly inspired by Dette et al. (2008), who propose an iteration of the form of (\ref{alg2+}) with a dynamic choice of $\alpha=\alpha^{(t)}$.  These authors provide an upper bound on $\alpha^{(t)}$ which ensures the monotonicity of (\ref{alg2+}), and also observe that their algorithm converges faster than Algorithm~I in numerical examples.  We shall also discuss the convergence rate of this dynamic algorithm.

Section~2 establishes the strict monotonicity of Algorithm~II.  The argument extends that of Yu (2010).  In Section~3, we analyze iteration (\ref{alg2+}) for fixed $\alpha$ in terms of both the {\it matrix rate} and the {\it global rate}.  For Algorithm~I (i.e., $\alpha=0$), it is shown that the matrix rate has only nonnegative eigenvalues.  Combined with a simple relation between the convergence rates of (\ref{alg2+}) for different $\alpha$, this shows that, with some exceptions, iteration (\ref{alg2+}) with $\alpha>0$ will converge faster than Algorithm~I.  Section~4 concludes with a small numerical illustration. 

\section{Strict Monotonicity of Algorithm~II}
Theorem 1 of Yu (2010) implies the monotonicity, but not {\it strict monotonicity}, of iteration (\ref{alg2}).  In Proposition \ref{prop1} below, we establish strict monotonicity for $m\geq 3$.  We include the proof of monotonicity for completeness, but our emphasis is on the equality condition.  Strict monotonicity plays a key role in the proof of the convergence theorem (Theorem \ref{thm0}). 

Let us denote $\Omega_+=\{w\in \bar{\Omega}:\ M(w)>0\ ({\rm positive\ definite})\}$.  In this section we assume $x_i=(1, z_i^\top)^\top$, and write $X\equiv (x_1, \ldots, x_n)^\top$. 
\begin{proposition}
\label{prop1}
Assume $m\geq 3$ and $X$ has full rank $m$.  Then iteration (\ref{alg2}) is strictly monotonic.  That is, if $w^{(t-1)}, w^{(t)}\in \Omega_+$ satisfy (\ref{alg2}), then $\det M(w^{(t-1)})\leq \det M(w^{(t)})$, with equality only if $w^{(t)}=w^{(t-1)}$. 
\end{proposition}

\begin{proof}
Let $K=(0_{m-1}, I_{m-1})$ where $0_r$ denotes the $r\times 1$ vector of zeros, and $I_r$ denotes the $r\times r$ identity matrix.  Define $\psi(M)=\log\det (KMK^\top)$ for any positive definite $m\times m$ matrix $M$.  Consider the function
$$h(\Sigma, w, Q)=\psi(\Sigma)+tr(\psi'(\Sigma)(Q\Delta^{-1}_wQ^\top-\Sigma))$$
where $\Sigma$ ($m\times m$) is positive definite, $w\in \Omega$, $\Delta_w={\rm Diag}(w)$, and $Q$ ($m\times n$) is full-rank.  Because $\psi(M)$ is concave in $M$, and strictly concave when restricted to $KMK^\top$, we have
\begin{equation}
\label{ineq1}
h(\Sigma, w, Q)\geq \psi(Q\Delta^{-1}_w Q^\top)
\end{equation}
with equality only when $K (Q\Delta^{-1}_w Q^\top-\Sigma) K^\top=0$. 

On the other hand, suppose $QX=I_m$, then we have 
\begin{equation}
\label{ineq2}
\psi(Q\Delta^{-1}_w Q^\top)\geq \psi(M^{-1}(w)).
\end{equation}
This holds because $QY$ is an unbiased estimator of $\theta$ in the linear model
$$Y\sim {\rm N}(X\theta,\, \Delta^{-1}_w).$$
Hence its variance matrix $Q\Delta^{-1}_w Q^\top$ is at least as large (in the positive definite ordering) as $(X^\top \Delta_w X)^{-1}=M^{-1}(w)$, which corresponds to the weighted least squares estimator $Q_{\rm WLS}=(X^\top \Delta_w X)^{-1} X^\top \Delta_w$.  Moreover, equality in (\ref{ineq2}) holds only when $K(Q-Q_{\rm WLS})=0$, i.e., when $QY$ agrees with $Q_{\rm WLS}Y$ in all coordinates except the first.  

Let $w^{(t-1)}, w^{(t)}\in \Omega$ be related by (\ref{alg2}).  Consider $Q^{(t-1)}=(X^\top \Delta_{w^{(t-1)}} X)^{-1} X^\top \Delta_{w^{(t-1)}}$ and define $Q^{(t)}$ similarly.  We have 
\begin{align}
\nonumber
\psi(M^{-1}(w^{(t-1)}))  &=h(M^{-1}(w^{(t-1)}), w^{(t-1)}, Q^{(t-1)})\\
\label{direct}
                         &=h(M^{-1}(w^{(t-1)}), w^{(t)}, Q^{(t-1)})\\
\label{step1}
                         &\geq \psi(Q^{(t-1)} \Delta_{w^{(t)}}^{-1} Q^{(t-1)\top})\\
\label{step2}
                         &\geq \psi(M^{-1}(w^{(t)})).
\end{align}
The key is the equality in (\ref{direct}), which follows from (\ref{alg2}) after some algebra.  The inequality (\ref{step1}) follows from (\ref{ineq1}).  The inequality (\ref{step2}) follows from (\ref{ineq2}).  We also have the easily verified identity $\psi(M^{-1}(w))=-\log\det M(w)$.  Thus the monotonicity statement holds. 

To prove strict monotonicity, let us check the equality conditions in (\ref{step1}) and (\ref{step2}).  The equality in (\ref{step1}) entails 
\begin{equation}
\label{eq2}
K(M^{-1}(w^{(t-1)})-Q^{(t-1)} \Delta_{w^{(t)}}^{-1} Q^{(t-1)\top})K^\top=0.
\end{equation}
The equality in (\ref{step2}) entails 
\begin{equation}
\label{eq3}
K(Q^{(t-1)}-Q^{(t)})=0, 
\end{equation}
which implies
$$K Q^{(t-1)}\Delta_{w^{(t)}}^{-1} Q^{(t-1)\top}K^\top =KQ^{(t)}\Delta_{w^{(t)}}^{-1} Q^{(t)\top} K^\top = KM^{-1}(w^{(t)})K^\top.$$
We obtain 
\begin{equation}
\label{eq4}
K(M^{-1}(w^{(t)})-M^{-1}(w^{(t-1)}))K^\top=0
\end{equation}
in view of (\ref{eq2}).  After some calculation, we can show that (\ref{eq3}) and (\ref{eq4}) imply  
$$X^\top (\Delta_{w^{(t-1)}} -w^{(t-1)}w^{(t-1)\top}) =X^\top (\Delta_{w^{(t)}} -w^{(t)}w^{(t)\top}).$$
Equivalently, 
\begin{equation}
\label{eq5}
w_i^{(t-1)}(x_i-\bar{x}^{(t-1)})=w_i^{(t)}(x_i-\bar{x}^{(t)}),\quad i=1,\ldots, n,
\end{equation}
where $\bar{x}^{(t)}=\sum_{i=1}^n w_i^{(t)} x_i$, and $\bar{x}^{(t-1)}$ is defined similarly.  If $w_i^{(t-1)}=w_i^{(t)}$ for any $i$, then $\bar{x}^{(t-1)}=\bar{x}^{(t)}\equiv \bar{x}$, and $(w_j^{(t-1)}-w_j^{(t)})(x_j-\bar{x})=0$ for all $j$.
That is, either $x_j-\bar{x}=0$, or $w_j^{(t-1)}=w_j^{(t)}$.  If $x_j-\bar{x}=0$, then $w_j^{(t)}=0$ by the form of (\ref{alg2}), which contradicts the assumption that $w^{(t)}\in \Omega$.  Hence, if $w_i^{(t-1)}=w_i^{(t)}$ for any $i$, then it holds for all $i$.  Let us assume $w_i^{(t-1)}\neq w_i^{(t)}$ for all $i$.  Rewriting (\ref{eq5}) we get 
\begin{equation}
\label{contra}
(\Delta_{w^{(t-1)}}-\Delta_{w^{(t)}}) X = (w^{(t-1)},\ -w^{(t)}) (\bar{x}^{(t-1)}, \bar{x}^{(t)})^\top.
\end{equation}
We obtain a contradiction because the left-hand side of (\ref{contra}) has rank $m\geq 3$, whereas the right-hand side has rank at most two.  It follows that $w^{(t)}=w^{(t-1)}$. 

Although the above argument assumes $w^{(t-1)}, w^{(t)}\in \Omega,$ i.e., they have all positive coordinates, the conclusion still holds if we only assume $w^{(t-1)}, w^{(t)}\in \Omega_+$.  First, we can restrict our analysis to the positive coordinates of $w^{(t-1)}$.  If $w^{(t-1)}\in \Omega$, but $w^{(t)}$ has some zero coordinates, then we can show $\det M(w^{(t-1)})<\det M(w^{(t)})$ by a limiting argument, upon close inspection of (\ref{direct}). 
\end{proof}

{\bf Remark.} When $m=2$, Algorithm~II is still monotonic, but may not be strictly monotonic; see Pronzato et al. (2000), Chapter 7, and Section~3 below. 

Strict monotonicity leads to the following convergence theorem, which fully resolves Titterington's (1978) conjecture. 

\begin{theorem}
\label{thm0}
Assume $m\geq 3$ and $X$ has full rank.  Let $w^{(t)}$ be a sequence generated by Algorithm~II, starting with $w^{(0)}\in \Omega$.  Then all limit points of $w^{(t)}$ are global maxima of $\det M(w)$ on $w\in \Omega_+$ and, as $t$ increases to $\infty$, $\det M(w^{(t)})$ increases to $\sup_{w\in \Omega_+} \det M(w)$.
\end{theorem}

Yu (2010, Theorem 2) presents a convergence theorem for a general class of multiplicative algorithms.  However, we cannot directly appeal to Theorem~2 in Yu (2010) because certain technical conditions are not satisfied.  For example, $w^{(t)}$ need not have all positive coordinates even if $w_i^{(t-1)}>0$ for all $i$.  However, inspection of (\ref{alg2}) shows that $w_i^{(t)}$ is set to zero only when $z_i=\bar{z}$, in which case it can be shown that an optimal design need not include $x_i$ as a support point, i.e., $x_i$ is safely eliminated.  Theorem~\ref{thm0} can then be proved by following the proof of Theorem~2 in Yu (2010) step by step (details omitted). 

\section{Rate of Convergence} 
In this section we analyze the convergence rate of iteration (\ref{alg2+}).  Assume the matrix $X=(x_1, \ldots, x_n)^\top$ has full rank $m\geq 2$.  Let $w^*\in \Omega$ be a global maximizer of $\det M(w)$.  We assume that $w^*$ has all positive components.  (A slightly weaker assumption is that the starting value $w^{(0)}$ has the same zero pattern as $w^*$.)  Though unrealistic in a practical situation, such an assumption makes our analysis tractable.  It seems to be a challenging problem to analyze the convergence rate when the algorithm tends to a boundary limit.  In Section~4, we present numerical examples to corroborate our rather idealized analysis. 
 
The notions of the matrix rate and the global rate are often used in analyzing fixed point algorithms in statistical contexts (Dempster et al. 1977; Meng, 1994).  Assume $0\leq \alpha<m$, and denote the mapping (\ref{alg2+}) by $T$.  The matrix rate of convergence of $T$ is defined as 
$$R(\alpha)=\left.\frac{\partial T(w)}{\partial w}\right|_{w=w^*},$$ 
because we have
\begin{equation}
\label{app}
T(w)-w^*\approx R(\alpha) (w-w^*)
\end{equation}
for $w$ near $w^*$.  The global rate of convergence, $r(\alpha)$, is defined as the spectral radius (the maximum modulus of the eigenvalues) of $R(\alpha)$ when restricted as a linear mapping on the space 
$$\Gamma=\left\{\gamma\in \mathbf{R}^m:\ 1_m \gamma=0\right\}$$
where $1_m$ denotes the $1\times m$ vector of ones.  Restricting $R(\alpha)$ to $\Gamma$ is possible because of the implication $\gamma\in \Gamma \Rightarrow R(\alpha)\gamma \in \Gamma$.  This restriction is imposed because we have the constraints $\sum_i w_i^*=\sum_i w_i=1$, and hence $w-w^*\in \Gamma$ in (\ref{app}).  Also note that such notions of convergence rates merely reflect how the iterations of $T$ behave near $w^*$; whether the algorithm converges from an arbitrary starting value is a different issue. 

Let us define $d_{ij}=x_i^\top M^{-1}(w^*) x_j,\ 1\leq i, j\leq n$.  The matrix rate $R(\alpha)$ admits a simple formula.

\begin{proposition}
The $(i, j)$th entry of $R(\alpha)$ is 
\begin{equation}
\label{entry}
R_{ij}(\alpha)=\left\{\begin{array}{ll} - w_i^* d_{ij}^2/(m-\alpha), & i\neq j,\\ (d_{ii}-w_i^* d_{ii}^2-\alpha)/(m-\alpha), & i=j.\end{array}\right.
\end{equation}
\end{proposition}
\begin{proof} By differentiating $I_m=M^{-1}(w) M(w)$ with respect to $w_j$ and rearranging, we obtain 
$$\frac{\partial M^{-1}(w)}{\partial w_j}=-M^{-1}(w)\frac{\partial M(w)}{\partial w_j} M^{-1}(w)=-M^{-1}(w) x_j x_j^\top M^{-1}(w),$$
which yields, for $i\neq j$, 
$$\frac{\partial T_i(w)}{\partial w_j}=\frac{w_i}{m-\alpha} x_i^\top \frac{\partial M^{-1}(w)}{\partial w_j}x_i=-\frac{w_i d_{ij}^2}{m-\alpha}.$$
The case of $i=j$ is similar.
\end{proof}

Theorem \ref{thm1} investigates the properties of $R(\alpha)$.  Its proof uses a lemma concerning the Hadamard product (see, e.g., Pukelsheim (1993), p. 199).
\begin{lemma}
\label{hadamard}
If $A=(A_{ij})$ and $B=(B_{ij})$ are symmetric nonnegative definite matrices of the same dimension, then their entry-wise product $C=((A_{ij}B_{ij}))$ is also nonnegative definite. 
\end{lemma}

\begin{theorem}
\label{thm1}
The matrix $R(\alpha)$ is diagonalizable, and all of its eigenvalues lie in the interval $[-\alpha/(m-\alpha),\, 1]$.
\end{theorem}
\begin{proof} 
We have $d_{ii}=m$ for all $i$ by the general equivalence theorem (Kiefer and Wolfowitz, 1960).  (Note the assumption that $w^*$ has all positive components.)  Thus 
\begin{equation}
\label{ra}
R(\alpha)=I_m-\frac{\Delta_{w^*} D^*}{m-\alpha},
\end{equation}
where $\Delta_{w^*}={\rm Diag}(w^*)$ as before, and $D^*=(d_{ij}^2)_{n\times n}$.  Define $D=(d_{ij})_{n\times n}$.  The formula $D=X M^{-1}(w^*) X^\top$ shows that $D$ is nonnegative definite, and so is $D^*$ by Lemma \ref{hadamard}, since $D^*$ is the entry-wise product of $D$ with itself.  By (\ref{ra}), $I_m-R(\alpha)$ is similar to $\Delta_{w^*}^{1/2} D^* \Delta_{w^*}^{1/2}/(m-\alpha)$, which is nonnegative definite because $D^*$ is.  Hence $I_m-R(\alpha)$ is diagonalizable, with nonnegative eigenvalues.  That is, $R(\alpha)$ is diagonalizable, and its eigenvalues do not exceed one.  On the other hand, $\Delta_{w^*} D^*$ satisfies (i) each entry is nonnegative, and (ii) each column sums to $\sum_i w_i^* d_{ij}^2=d_{jj}=m$.  By the Frobenius-Perron theorem (see Horn and Johnson (1990), Chapter 8), any eigenvalue of $\Delta_{w^*} D^*$ cannot exceed $m$.  It follows from (\ref{ra}) that any eigenvalue of $R(\alpha)$ is at least $1-m/(m-\alpha)=-\alpha/(m-\alpha)$.
\end{proof}

For Algorithm~I, Theorem \ref{thm1} leads to an eigenvalue bound similar to that for the EM algorithm; see Yu (2009) for another similar situation in the context of Shannon theory. 

\begin{corollary}
\label{coro1}
All eigenvalues of $R(0)$ lie in the interval $[0, 1]$. 
\end{corollary}

Because Algorithm~I converges, it is not surprising that eigenvalues of $R(0)$ do not exceed one.  However, that these eigenvalues are nonnegative shows that the iterations of Algorithm~I are conservative, and may be improved by some form of overrelaxation, e.g., by using $\alpha>0$ in (\ref{alg2+}).  To make this intuition precise, we compare convergence rates of iteration (\ref{alg2+}) for different $\alpha$ (with respect to the same $w^*$).  Equation (\ref{ra}) yields 
\begin{equation}
\label{speed}
I_m-R(\alpha)=\frac{m}{m-\alpha} (I_m-R(0)).
\end{equation}
If we define $I_m-R(\alpha)$ as the {\it matrix speed} of convergence, then (\ref{speed}) has an appealing interpretation: iteration (\ref{alg2+}) is precisely $m/(m-\alpha)$ times as fast as Algorithm~I.  Nevertheless, one should be cautious toward such an interpretation.  First, we need to assume that (\ref{alg2+}) converges, which is not always guaranteed.  Secondly, when some of the eigenvalues of $R(\alpha)$ are negative, iteration (\ref{alg2+}) can actually be slower than Algorithm~I, as the following example illustrates.  Let $n=m=2$ and consider the design space $\mathcal{X}=\{x_1=(1, -1)^\top,\ x_2=(1, 1)^\top\}$.  Iteration (\ref{alg2+}) maps any $w^{(t-1)}=(w_1, w_2)^\top\in \Omega$ to 
$$w^{(t)}=\frac{1}{2-\alpha}(1-\alpha w_1, 1-\alpha w_2)^\top.$$
We have $R(\alpha)=-\alpha I_2/(2-\alpha)$.  The algorithm reaches $w^*=(1/2, 1/2)^\top$ in one iteration if $\alpha=0$, but becomes slower and slower as $\alpha$ increases from $0$ to $1$.  When $\alpha=1$ it alternates between two points $(w_1, w_2)^\top$ and $(w_2, w_1)^\top$ (assuming $w_1\neq w_2$) and does not even converge.  (Non-convergence of Algorithm~II when $m=2$ has been noted by Pronzato et al. (2000).)

What (\ref{speed}) does imply is that, if $\alpha$ is not too large, and if Algorithm~I itself is slow, then iteration (\ref{alg2+}) will converge faster than Algorithm~I.  Intuitively, an $\alpha$ too large would overshoot and slow the algorithm down.  Proposition \ref{prop2} makes this explicit by comparing the global rate $r(\alpha)$. 

\begin{proposition}
\label{prop2}
Assume $r(0)\geq 2\alpha/m$.  Then 
\begin{equation}
\label{precise}
1-r(\alpha)= \frac{m}{m-\alpha} (1-r(0)),
\end{equation}
and hence $r(\alpha)\leq r(0)$.
\end{proposition}
\begin{proof}
Let $r_+(\alpha)$ (resp.\ $r_-(\alpha)$) denote the largest (resp.\ smallest) eigenvalues of $R(\alpha)$ when restricted as a linear mapping on $\Gamma$.  Then $r(\alpha)=\max\{|r_+(\alpha)|,\, |r_-(\alpha)|\}$.  Corollary \ref{coro1} implies $r(0)=r_+(0)$.  By (\ref{speed}), $$1-r_+(\alpha) = \frac{m}{m-\alpha} (1-r(0))\leq \frac{m-2\alpha}{m-\alpha}.$$  
That is, $r_+(\alpha)\geq \alpha/(m-\alpha)$.  On the other hand, Theorem \ref{thm1} implies $r_-(\alpha)\geq -\alpha/(m-\alpha)$.  Hence $r_+(\alpha)\geq |r_-(\alpha)|$, and $r(\alpha)=r_+(\alpha)$, thus proving (\ref{precise}).
\end{proof} 

\begin{corollary}
\label{coro}
If $r(0)\geq 2/m$, then the global rate of Algorithm~II is no worse than that of Algorithm~I. 
\end{corollary}

Corollary~\ref{coro} suggests that Algorithm~II is likely to converge faster than Algorithm~I, as long as $m\geq 3$ and 
$r(0)$ is reasonably large.  Note that in a practical situation, Algorithm~I can be quite slow, i.e., $r(0)$ is close to one.  This explains the observed improvement of using Algorithm~II in numerical examples. 

Dette et al.\ (2008) consider a version of (\ref{alg2+}) where $\alpha=\alpha^{(t)}$ is set at each iteration.  It is shown that by choosing 
\begin{equation}
\label{dynamic}
\alpha^{(t)}=\frac{1}{2}\min_i x_i^\top M^{-1}(w^{(t-1)}) x_i
\end{equation}
the resulting algorithm is monotonic, and usually converges faster than Algorithm~I.  Although this algorithm is dynamic, we can still discuss its asymptotic rate of convergence, because if $t\to\infty$ and $w^{(t)}\to w^*$, then $\alpha^{(t)}$ also tends to a limit: 
$$\hat{\alpha}\equiv \lim_{t\to\infty} \alpha^{(t)}=\frac{1}{2}\min_i x_i^\top M^{-1}(w^*) x_i.$$
It follows that, for large $t$, each iteration of this dynamic algorithm behaves as if $\alpha$ is fixed at $\hat{\alpha}$.  If $w^*$ has all positive components, then $\hat{\alpha}=m/2$ by the general equivalence theorem; in general $0\leq \hat{\alpha}\leq m/2$.  If $\hat{\alpha}=m/2$, and if (\ref{precise}) holds, then we can loosely say that the dynamic algorithm is twice ($m/(m-\hat{\alpha})=2$) as fast as Algorithm~I.  In a practical situation, however, it is more likely that $\hat{\alpha}<m/2$, hence we may expect a less pronounced improvement; see Section~4 for a numerical example. 

\section{Numerical Example}
The formula (\ref{precise}) is derived under the assumption that all coordinates of $w^*$ are positive.  As mentioned earlier, in realistic problems this is usually not true.  It is therefore reasonable to ask whether (\ref{precise}) holds in any practical sense.  To study this, we employ an empirical measure of the convergence rate, defined as 
\begin{equation}
\label{empirical}
\hat{r}=\lim_{t\to\infty} \frac{|w^{(t+1)}-w^{(t)}|}{|w^{(t)}-w^{(t-1)}|},
\end{equation}
where $|v|=(\sum_i v_i^2)^{1/2}$.  We compare the $\hat{r}$ for iteration (\ref{alg2+}) with different values of $\alpha$ for a few regression models.  Define $s_i=i/20,\ i=1, \ldots, 20$.  Similar to Dette et al.\ (2008), we consider design spaces 
\begin{align*}
\mathcal{X}_1 &=\{x_i=(1,\, e^{-s_i},\, s_i e^{-s_i})^\top,\ i=1,\ldots, 20\};\\
\mathcal{X}_2 &=\{x_i=(1,\, s_i/(\kappa+s_i),\, s_i/(\kappa+s_i)^2)^\top,\ i=1, \ldots, 20\},\quad \kappa=0.5;\\
\mathcal{X}_3 &=\{x_i=(1,\, s_i,\, s_i^2,\, s_i^3)^\top,\ i=1,\ldots, 20\}.
\end{align*}
Note that a D-optimal design on $\mathcal{X}_2$ is equivalently a locally D-optimal design for the parameter $(\beta_0, \beta_1, \kappa)$ in the nonlinear model,
\begin{equation}
\label{mm}
y=\beta_0+\frac{\beta_1 s}{\kappa+s} +\epsilon,\quad \epsilon\sim {\rm N}(0,\, \sigma^2),
\end{equation}
where the design space for $s$ is $\{s_i=i/20,\ i=1,\ldots, 20\}$, and the prior guess for $\kappa$ is $\kappa^*=0.5$.  
Mathematically, (\ref{mm}) with $\beta_0=0$ corresponds to the Michaelis-Menten model often employed to describe enzyme kinetics. 

Table~1 records the estimates of $1-\hat{r}$ for iteration (\ref{alg2+}) with various choices of $\alpha$ (fixed or dynamic).  Each algorithm is started from the uniform design ($w_i=1/20,\ i=1,\ldots, 20$), and $\hat{r}$ is estimated by the ratio on the right hand side of (\ref{empirical}) when it stabilizes.  The first three columns of Table 1 deal with fixed $\alpha$, in which case we write $\hat{r}=\hat{r}(\alpha)$.  If we interpret $1-\hat{r}$ as the {\it empirical speed} of convergence, then evidently larger values of $\alpha$ improve the speed.  For $\mathcal{X}_1$ and $\mathcal{X}_2$, the ratio of improvement, $(1-\hat{r}(\alpha))/(1-\hat{r}(0)),$ is approximately equal to $m/(m-\alpha),\ \alpha=0.5, 1$.  For $\mathcal{X}_3$, this ratio is below the value suggested by (\ref{precise}) for either $\alpha=0.5$ or $\alpha=1$.  However, it is possible that accurate estimation of the ratio of improvement becomes more difficult because the algorithms are much slower for $\mathcal{X}_3$ than for $\mathcal{X}_1$ or $\mathcal{X}_2$.  The last column concerns the algorithm of Dette et al.\ (2008) where $\alpha$ is set dynamically as in  (\ref{dynamic}).  The limiting value $\lim_{t\to\infty} \alpha^{(t)}$ is estimated at $\hat{\alpha}_1=0.939$ for $\mathcal{X}_1$, $\hat{\alpha}_2=0.935$ for $\mathcal{X}_2$, and $\hat{\alpha}_3=1.303$ for $\mathcal{X}_3$.  We observe that these agree well with what (\ref{precise}) suggests.  For example, the ratios of improvement for the dynamic algorithm are 
$$\frac{0.0245}{0.0168}\approx \frac{3}{3-\hat{\alpha}_1}\quad {\rm and}\quad \frac{0.0256}{0.0177}\approx \frac{3}{3-\hat{\alpha}_2}$$
for $\mathcal{X}_1$ and $\mathcal{X}_2$ respectively.  Overall, we believe that (\ref{precise}) remains suggestive of how much iteration (\ref{alg2+}) can improve upon Algorithm~I in realistic situations. 

\begin{table}
\caption{Values of $1-\hat{r}$ (the empirical speed) for iteration (\ref{alg2+}) with several design spaces and choices of $\alpha$.  Dynamic $\alpha$ refers to the algorithm of Dette et al. (2008).}
\begin{center}
\begin{tabular}{l|rrrr}
\hline & $\alpha=0$ & $\alpha=0.5$ & $\alpha=1$ & dynamic $\alpha$\\
\hline
$\mathcal{X}_1$ & 0.0168 & 0.0202  & 0.0252 & 0.0245\\
$\mathcal{X}_2$ & 0.0177 & 0.0212  & 0.0264 & 0.0256\\
$\mathcal{X}_3$ & 0.0062 & 0.0068  & 0.0076 & 0.0082\\
\hline 
\end{tabular}
\end{center}
\end{table}

\section*{Acknowledgments}
The author would like to thank Don Rubin, Xiao-Li Meng, and David van Dyk for introducing him to the field of statistical computing.  He is also grateful to the editor, the associate editor, and the referees for their valuable comments.

\end{document}